\documentclass[12pt, english]{article} 
\usepackage{amssymb,amsmath,amsfonts,color,mathrsfs,amsthm,verbatim}
\makeatletter

\newtheorem {theorem}{Theorem}[section]
\newtheorem {proposition}{Proposition}[section]

\newtheorem {example}{Example}[section]

\newcommand{\hrow}[1]{{H_{\mathrm{row}(#1)}}}
\newcommand{\srow}[1]{{S_{\mathrm{row}(#1)}}}
\newcommand{\scol}[1]{{S_{\mathrm{col}(#1)}}}
\newcommand{\col}{{\mathrm{col}}}

\newcommand{\supp}[1]{{\mathrm{supp}(\boldsymbol{#1})}}
\newcommand{\wt}[1]{{\mathrm{wt}(\boldsymbol{#1})}}
\newcommand{\wtminus}[1]{{\mathrm{wt}(-\boldsymbol{#1})}}
\newcommand{\dist}[2]{{\mathrm{d}(\boldsymbol{#1}, \boldsymbol{#2})}}
\newcommand{\mindist}[1]{{\mathrm{d}(#1)}}
\newcommand{\cdist}[1]{{\mathrm{d}(\boldsymbol{#1}, C)}}

\newcommand{\decomp}{{C = \bigcup_{j = 1}^s (\boldsymbol{d}_j + D_C)}}
\newcommand{\decomparg}[1]{{C = \bigcup_{j = 1}^#1 (\boldsymbol{d}_j + D_C)}}

\newcommand{\sdist}{{S^{\mathrm{diff}}}}
\newcommand{\srowelem}[1]{{S_{\boldsymbol{#1}}}}

\newcommand{\echar}[1]{{\boldsymbol{\varepsilon}_{\boldsymbol{#1}}}}

\newcommand{\echararg}[2]{{\boldsymbol{\varepsilon}_{\boldsymbol{#1}}(\boldsymbol{#2})}}
\newcommand{\echarargindex}[3]{{\boldsymbol{\varepsilon}_{\boldsymbol{#1}}(\boldsymbol{#2}_{#3})}}
\newcommand{\echarargminus}[2]{{\boldsymbol{\varepsilon}_{\boldsymbol{#1}}(-\boldsymbol{#2})}}
\newcommand{\echarargminusindex}[3]{{\boldsymbol{\varepsilon}_{\boldsymbol{#1}}(-\boldsymbol{#2}_{#3})}}
\newcommand{\echarargminuscoset}[3]{{\boldsymbol{\varepsilon}_{\boldsymbol{#1}}(-\boldsymbol{#2}_{#3} + \boldsymbol{d})}}

\newcommand{\echartwo}[2]{{\boldsymbol{\varepsilon}_{(\boldsymbol{#1},\boldsymbol{#2})}}}

\begin{document}
\pagestyle{plain}

\begin{center}
\begin{Huge}
Parity check systems of nonlinear codes over finite commutative Frobenius rings\\
\end{Huge} 
Thomas Westerb\"ack \footnote{Supported by grant KAW 2005.0098 from the Knut and Alice Wallenberg Foundation}, \\
Department of Mathematics, KTH, \\
S-100 44 Stockholm, Sweden\\
thowest@math.kth.se\\
\end{center} 

\begin{abstract}
The concept of parity check matrices of linear binary codes has been extended by Heden \cite{heden08a} to parity check systems of nonlinear binary codes. In the present paper we extend this concept to parity check systems of nonlinear codes over finite commutative Frobenius rings. Using parity check systems, results on how to get some fundamental properties of the codes are given. Moreover, parity check systems and its connection to characters is investigated and a MacWilliams type theorem on the distance distribution is given. 
\end{abstract}



\section{Introduction}

In the 1990's Nechaev \cite{nechaev91} and independently Hammons et al. \cite{hammons94} discovered  that several families of nonlinear binary codes with good parameters, regarding the number of codewords and the minimum Hamming distance, can be represented as linear codes over $\mathbb{Z}_4$. Thereafter there has been a revival in the study of codes over finite rings. For some purposes finite Frobenius rings seem to be the most appropriate rings to use for codes over rings, see for example \cite{wood99, wood08}.

In \cite{heden08a}, Heden generalizes the concept of parity check matrices for linear binary codes to parity check systems for nonlinear binary codes. A parity check system $(H | S)$ is a concatenation of two matrices $H$ and $S$, see \eqref{def:parity check system}. In \cite{heden08a} there is a sufficient and necessary condition for a parity check system to correspond to a perfect 1-error correcting binary code. This condition has then been a fruitful approach to the study of perfect binary codes, see for example \cite{heden06, heden08a, hessler06, villanueva09}. By using this condition, the last remaining open case for the rank-kernel problem for binary perfect codes was solved in \cite{heden06}. This problem was given in \cite{etzion98}. Parity check systems have also been used in a software package for Magma, see \cite{villanueva09}, in order to represent and construct nonlinear perfect binary codes in an efficient way. Further, in \cite{heden08b}, parity check systems are defined for perfect codes over finite fields of other cardinalities than $2$, namely for all primes. 

In \cite{villanueva14}, by the use of parity check systems, results on how to efficiently represent, manipulate, store and construct nonlinear binary codes are given. The technique of using parity check systems in \cite{villanueva14} is especially suitable for codes with large kernel. (A code $C$ that consists of $s$ cosets of a subspace has a large kernel when $s << |C|$. For the definition of a kernel see \eqref{eq:def_ker}.)  Also, algorithms on how to compute the minimum distance of the codes and how to decode them are given in \cite{villanueva14}. A comparison of the performance of these algorithms compared to some well-known algorithms and a brute force method for some classes of nonlinear binary codes are given. When the kernel is large enough the algorithms developed in \cite{villanueva14} performs best. 

Two examples of families of good nonlinear binary codes with large kernel are the Preparata and Kerdock codes of length $2^l$ for even $l \geq 4$. A Kerdock code $C$ consists of $2^{l-1}$ cosets of a subspace and $|C| = 2^{2^l}$. A Preparata code $C$ consists of $2^{l-1}$ cosets of a subspace and $|C| = 2^{2^l - 2l}$. For more details on these families of codes see \cite{macwilliams77}.

One of the most fundamental results in coding theory is a theorem due to MacWilliams \cite{macwilliams63}, which relates the weights of a linear code in a finite vector space to the weights of its dual code. There are many generalizations of this result, see for example \cite{britz02, delsarte72, greferath04, macwilliams77, wood99}. By the use of characters, defined in \eqref{eq:character_mult}, parity check matrices and a  MacWilliams type theorem is defined for linear codes over any finite ring in \cite{greferath04}.

In this paper we extend the concept of parity check systems for nonlinear binary codes, introduced in \cite{heden08a}, to parity check systems for codes over finite commutative Frobenius rings. (See the remark in the end of Section \ref{sec:parity check} for more details about parity check systems for codes over finite fields.) The definition of parity check system is given in \eqref{def:parity check system} and the fundamental connections between parity check systems and codes are given in Theorem \ref{theorem:code to parity check system} and Theorem \ref{theorem:parity check system to code}. Considering the Hamming distance, by the use of a parity check system of a code, Theorem \ref{theorem:d(C)} shows how to derive the minimum distance and Theorem \ref{theorem:error-correct} how to error-correct. 

By use of the module isomorphism given in \eqref{eq:module_iso}, for any submodule $D$ of $R^n$, we are able in \eqref{eq:lozenge_perp} to identify the dot product dual $D^\perp$ in \eqref{eq:def_standard_dot_dual} with the character dual $D^\lozenge$ in \eqref{eq:def_character_dual}. By the use of this identification, Theorem \ref{theorem:Fourier parity check system} gives a formula on how a parity check system $(H|S)$ of a code $C$ can be used in order to get the Fourier coefficients $\hat{\delta}_C$, defined in \eqref{eq:fourier transform M} and \eqref{eq:def_fourier_repr_C}. This formula used in Theorem \ref{theorem:distance identity} gives a MacWilliams type of theorem on the distance distribution of the code using a parity check system. Any parity check matrix of a linear code over a finite commutative Frobenius ring, as defined in \cite{greferath04}, corresponds to a parity check system over the code. Using this fact,  the MacWilliams identity  given in \cite{greferath04} corresponds to Theorem \ref{theorem:distance identity} for linear codes over finite commutative Frobenius rings. For more details about these correspondences see the remark after Theorem \ref{theorem:distance identity}. 

The main contributions in this paper are; the extension of the concept of binary parity check systems to parity check systems over finite commutative Frobenius rings, some fundamental results on the connections between codes, parity check systems and characters, and results concerning the minimal distance and distance distribution of codes using parity check systems. The main motive for the present work is to give fundamental results on parity check systems such that they can be used in further research on nonlinear codes over finite commutative Frobenius rings. As described above, earlier works have shown that parity check systems can be used fruitfully in order to do research on nonlinear binary codes. In the line of these earlier works, some interesting areas for future studies on nonlinear codes over commutative Frobenius rings, using parity check systems, are; how to characterize and investigate different classes of codes, and how to efficiently represent, manipulate, store and construct codes with good properties. 

The paper is organized as follows. Section \ref{section:preliminaries} contains some basic facts and notation on finite commutative Frobenius rings $R$, codes over $R$ and parity check systems over $R$.  Section \ref{sec:parity check} deals with the fundamental connection between parity check systems and codes. The first part of Section \ref{section:dfa} describes some basic concepts and facts about Fourier analysis on finite Abelian groups. In the second part some results on Fourier analysis on codes over $R$ are given. Section \ref{sec:distance} mainly deals with how to get the minimal distance and distance distribution of codes by the use of parity check systems, and how parity check systems are connected to characters. 


\section{Preliminaries} \label{section:preliminaries}

We assume that the reader is familiar with standard terminology in ring theory, module theory and coding theory, for more details, see for example \cite{macwilliams77, sharp00}. For further reading on Frobenius rings and the application of these rings to coding theory, see for example \cite{greferath09, greferath04, greferath00, honold01, honold99, wood99, wood08}.

Let $A$ denote a finite commutative (associative) ring with identity ($1_A \neq 0_A$). For any subset $I \subseteq A$ let $I^\perp$ denote the \emph{annihilator} of $I$, i.e. 
$$
I^\perp = \{a \in A : xa = 0 \hbox{ for all } x \in I\}.
$$
An $A$-module structure on $A^n$ is obtained by
$$
\boldsymbol{x} + \boldsymbol{y} = (x_1 + y_1,\ldots,x_n + y_n) \hbox{ and } a \boldsymbol{x} = (ax_1,\ldots,ax_n),
$$ 
for any $a \in A$ and $\boldsymbol{x} = (x_1,\ldots,x_n) \hbox{, } \boldsymbol{y}=(y_1,\ldots,y_n)\in A^n$. A dot product for $\boldsymbol{x}, \boldsymbol{y} \in A^n$ is defined by 
$$
\boldsymbol{x} \cdot \boldsymbol{y} = x_1 y_1 + \ldots + x_n y_n.
$$
For any subset $B \subseteq A^n$, let 
\begin{equation} \label{eq:def_standard_dot_dual}
B^\perp = \{\boldsymbol{y} \in A^n | \boldsymbol{x} \cdot \boldsymbol{y} = 0 \hbox{ for all } \boldsymbol{x} \in B\}.
\end{equation}

There are many equivalent characterizations of Frobenius rings \cite{honold01}. One characterization of finite commutative Frobenius ring is that $A$ is \emph{Frobenius} if and only if
$$
|I| \cdot |I^\perp| = |A| \hbox{ for every ideal } I \hbox{ of } A.
$$
Examples of finite commutative Frobenius rings are finite commutative principal ideal rings and finite direct sums of finite commutative Frobenius rings. Let us also mention that examples of finite commutative principal ideal rings are $\mathbb{F}_q$ (the finite field of cardinality $q$), the ring of integers  $\mathbb{Z}_t$ modulo $t$ , Galois rings, finite commutative chain rings and finite direct sums of finite commutative principal ideal rings, see for example \cite{greferath09}.

Henceforth, let $R$ be a finite commutative Frobenius ring. If $D$ is a submodule of $R^n$, then $D^\perp$ is a submodule of $D$ with the following to properties
\begin{equation} \label{eq:equality_D_Ddualdual}
D^{\perp \perp} = D
\end{equation}
and
\begin{equation} \label{eq:cardinality_Ddual_frobenius}
|D| \cdot | D^\perp | = |R|^n,
\end{equation}
see for example \cite{honold01a}.

Let $(H | S)$ be a concatenation of a $m \times n$-matrix $H = (h_{i,j})$ and a $m \times s$-matrix $S = (s_{i,j})$ over $R$. Moreover, let $\hrow{i}$ denote row $i$ of $H$, $\srow{i}$ denote row $i$ of $S$, $\scol{j}$ denote row $j$ of $S$ and $\col(S)$ denote the family of columns of $S$. The concatenated matrix $(H | S)$ is a \emph{parity check system over $R$} if the following conditions are satisfied,
\begin{equation} \label{def:parity check system}
\begin{array}{rl}
(i)   & s_{i,j} \in \{\hrow{i} \cdot \boldsymbol{x} : \boldsymbol{x} \in R^n \} \hbox{ for } 1 \leq i \leq m \hbox{ and } 1 \leq j \leq s,\\
(ii)  & \hbox{all columns in $S$ are distinct,}\\
(iii) & \sum_{i = 1}^m r_i \hrow{i} =  \sum_{i = 1}^m r'_i \hrow{i} \hbox{ for } r_i, r'_i \in R \Rightarrow \\
      & \sum_{i = 1}^m r_i \srow{i} =  \sum_{i = 1}^m r'_i \srow{i} 
\end{array}
\end{equation}
Henceforth, let $m$, $n$ and $s$ denote the size of a parity check system $(H | S)$, as indicated above. 

\begin{example}
An example of a parity check system $(H|S)$ over $\mathbb{Z}_6$, with $m = 2$, $n= 4$ and $s=3$, is
$$
(H|S) = 
\left (
\begin{array}{cccc|ccc}
1&1&3&5&0&1&5\\
0&4&2&2&0&2&4
\end{array}
\right )
.
$$
Condition (i) in \eqref{def:parity check system} is satisfied as 
$$
\{(1,1,3,5) \cdot \boldsymbol{x} : \boldsymbol{x} \in \mathbb{Z}_6^4\} = \mathbb{Z}_6 \hbox{ and } \{(0,4,2,2) \cdot \boldsymbol{x} : \boldsymbol{x} \in \mathbb{Z}_6^4\} = \{0,2,4\}.
$$
Condition (ii) is satisfied as the columns in $S$, $(0,0)^T$, $(1,2)^T$ and $(5,4)^T$, are distinct. For matrix $H$ we have that 
$$
r_1(1,1,3,5) + r_2 (0,4,2,2) = r'_1(1,1,3,5) + r'_2 (0,4,2,2) \Rightarrow r'_1 = r_1 \hbox{ and } r'_2 = r_2,
$$
for $r_1,r_2,r'_1,r'_2 \in \mathbb{Z}_6$. Consequently, condition (iii) is satisfied by $(H|S)$. 
\end{example}

For $\boldsymbol{x}, \boldsymbol{y} \in R^n$, let $\supp{x}$, $\wt{x}$ and $\dist{x}{y}$ denote the \emph{support} of $\boldsymbol{x}$, \emph{Hamming weight} of $\boldsymbol{x}$ and \emph{Hamming distance} between $\boldsymbol{x}$ and $\boldsymbol{y}$, respectively. That is
\begin{equation} \label{eq:min_supp_weight}
\begin{array}{l}
\supp{x}  = \{i \in [n] : x_i \neq 0\} \hbox{, } \wt{x} = | \mathrm{supp(\boldsymbol{x})} | \hbox{ and } \dist{x}{y} = \mathrm{wt(\boldsymbol{x}-\boldsymbol{y})},
\end{array}
\end{equation}
where $[n] = \{1,\ldots,n\}$. A \emph{code} $C$ over $R$ is a nonempty subset of $R^n$. The elements of $C$ are called \emph{codewords} and a code is \emph{linear} if it is a submodule of $R^n$, otherwise the code is \emph{nonlinear}. For any code $C$ over $R$ let $\mindist{C}$ denote the minimum distance of $C$, i.e. 
\begin{equation} \label{eq:mindist}
\mindist{C} = \min \{ \dist{x}{y} : \boldsymbol{x}, \boldsymbol{y} \in C\ \hbox{ and } \boldsymbol{x} \neq \boldsymbol{y} \}.
\end{equation}
Moreover, for any code $C$ of $R^n$ and element $\boldsymbol{x} \in R^n$, let
$$
\boldsymbol{x} + C = \{\boldsymbol{x} + \boldsymbol{c} : \boldsymbol{c} \in C\}.
$$
The \emph{kernel} of a code $C$ of $R^n$ is the following submodule of $R^n$,
\begin{equation} \label{eq:def_ker}
\mathrm{ker}(C) = \{\boldsymbol{x} \in R^n : r \boldsymbol{x} + C = C \hbox{ for all } r \in R\}.
\end{equation}
A \emph{partial kernel} of $C$ is a submodule of $\mathrm{ker}(C)$. For any partial kernel $D$ of $\mathrm{ker}(C)$, there are elements $\boldsymbol{d}_1,\ldots,\boldsymbol{d}_s \in C$ such that 
\begin{equation} \label{eq:coset decomposition}
C = \bigcup_{j = 1}^s (\boldsymbol{d}_j + D) \quad \hbox{where} \quad (\boldsymbol{d}_i + D) \cap (\boldsymbol{d}_j + D) = \emptyset \hbox{ if } i \neq j.
\end{equation}
Note, that if $D$ is a submodule of $R^n$ such that (\ref{eq:coset decomposition}) holds for some $\boldsymbol{d}_1,\ldots,\boldsymbol{d}_s \in C$, then $D$ is a partial kernel of $C$. A decomposition, as in (\ref{eq:coset decomposition}), is here called a \emph{coset decomposition} of $C$.

For the rest of this paper, we assume that $C$ is a code of $R^n$ with a partial kernel $D_C$ and \emph{coset representatives} $\boldsymbol{d}_1,\ldots,\boldsymbol{d}_s$, i.e. $C$ has a coset decomposition 
$$
\decomp.
$$ 


\section{Parity check systems and codes over $R$} \label{sec:parity check}

The two theorems below give the fundamental connection between parity check systems and codes over $R$. The proofs of the theorems are given in this section after Proposition \ref{proposition:cosets}. Some more basic results and observations on parity check systems and codes over $R$ are given in the end of this section. 

Let $<\boldsymbol{x}_1,\ldots, \boldsymbol{x}_t>$ be the submodule of $R^n$ that is generated by the elements $\boldsymbol{x}_1,\ldots, \boldsymbol{x}_t \in R^n$.

\begin{theorem} \label{theorem:code to parity check system}
To any set of generators $<\boldsymbol{h}_1,\ldots,\boldsymbol{h}_m> = D_C^\perp$ for a code $\decomp$, the concatenated $m \times (n+s)$-matrix
\begin{equation} \label{eq:matrix representation}
\left (
\begin{array}{c|ccc}
\boldsymbol{h}_1 & \boldsymbol{h}_1 \cdot \boldsymbol{d}_1 & \cdots & \boldsymbol{h}_1 \cdot \boldsymbol{d}_s\\
\vdots & \vdots             & \ddots & \vdots\\
\boldsymbol{h}_m & \boldsymbol{h}_m \cdot \boldsymbol{d}_1 & \cdots & \boldsymbol{h}_m \cdot \boldsymbol{d}_s
\end{array}
\right ).
\end{equation}
is a parity check system over $R$.
\end{theorem}

\begin{example} \label{ex:C_to_PC}
Let $\decomparg{3}$ be a code over $\mathbb{Z}_6$, where 
$$
\begin{array}{l}
\boldsymbol{d}_1 = (0,0,0,0) \hbox{, } \boldsymbol{d}_2 = (5,2,0,0) \hbox{, } \boldsymbol{d}_3 = (4,1,0,0) \hbox{ and }\\
D_C = <(2,1,1,0),(0,1,0,1),(3,0,3,0)>.
\end{array}
$$
Then $D_C^\perp = <\boldsymbol{h}_1, \boldsymbol{h}_2,> = <(1,1,3,5),(0,4,2,2)>$ and
$$
(H|S) = 
\left (
\begin{array}{c|ccc}
\boldsymbol{h}_1 & \boldsymbol{h}_1 \cdot \boldsymbol{d}_1 & \boldsymbol{h}_1 \cdot \boldsymbol{d}_2 & \boldsymbol{h}_1 \cdot \boldsymbol{d}_3\\
\boldsymbol{h}_1 & \boldsymbol{h}_2 \cdot \boldsymbol{d}_1 & \boldsymbol{h}_2 \cdot \boldsymbol{d}_2 & \boldsymbol{h}_2 \cdot \boldsymbol{d}_3
\end{array}
\right ) =
\left (
\begin{array}{cccc|ccc}
1&1&3&5&0&1&5\\
0&4&2&2&0&2&4
\end{array}
\right )
$$
is a parity check system over $\mathbb{Z}_6$.
\end{example}

\begin{theorem} \label{theorem:parity check system to code}
Let $(H|S)$ be a parity check system over $R$. Then there is a unique code $C=$ $\bigcup_{j = 1}^s (\boldsymbol{d}_j + D_C)$ such that 
$$
\begin{array}{rl}
(i) & D_C = <\hrow{1},\ldots,\hrow{m}>^\perp,\\
(ii) & \boldsymbol{d}_j + D_C = \{ \boldsymbol{x} \in R^n : H \boldsymbol{x}^T = \scol{j} \} \hbox{ for } 1 \leq j \leq s.
\end{array}
$$
\end{theorem}

\begin{example}
Let $(H|S)$ be the following parity check system over $\mathbb{Z}_6$,
$$
(H|S) = 
\left (
\begin{array}{cccc|ccc}
1&1&3&5&0&1&5\\
0&4&2&2&0&2&4
\end{array}
\right ).
$$
Then $\decomparg{3}$ is the code over $\mathbb{Z}_6$ with 
$$
\begin{array}{l}
\boldsymbol{d}_1 + D_C = \{\boldsymbol{x} \in \mathbb{Z}_6^4: H \boldsymbol{x}^T = (0,0)^T\} = D_C,\\
\boldsymbol{d}_2 + D_C = \{\boldsymbol{x} \in \mathbb{Z}_6^4: H \boldsymbol{x}^T = (1,2)^T\} = (5,2,0,0) + D_C,\\
\boldsymbol{d}_3 + D_C = \{\boldsymbol{x} \in \mathbb{Z}_6^4: H \boldsymbol{x}^T = (5,4)^T\} = (4,1,0,0) + D_C,
\end{array}
$$
where $D_C = <(2,1,1,0),(0,1,0,1),(3,0,3,0)> = <(1,1,3,5), (0,4,2,2)>^\perp$.
\end{example}

We will say that a parity check system associated to a code $C$, as described in Theorem \ref{theorem:code to parity check system} and Theorem \ref{theorem:parity check system to code} above, is a \emph{parity check system of $C$}. Note, that every parity check system uniquely represent a code $C$, but a code $C$ can be represented by many different parity check systems. 

For a parity check system $(H |S)$ and its associated code $\decomp$, 
$$
\boldsymbol{x} \in C \iff H \cdot \boldsymbol{x}^T \in \col(S)
$$
for $\boldsymbol{x} \in R^n$. Hence, the complexity of checking if an element in $R^n$ is an element in the code or not are proportional to the number of columns in $S$. That is, we may have to check all the $s$ columns in $S$ where $s = \frac{|C|}{|D_C|}$. Thus, to represent a nonlinear code with a parity check system is most efficient when $D_C$ is large.  The representation of a code via a parity check system can be very inefficient when $|D_C| << |C|$. For example, if $|D_C| = 1$ then $s=|C|$.

The following proposition will be used in the proofs of Theorem \ref{theorem:code to parity check system} and Theorem \ref{theorem:parity check system to code}.

\begin{proposition} \label{proposition:cosets}
If $D$ is a submodule of $R^n$ and $<\boldsymbol{h}_1,\ldots,\boldsymbol{h}_m> = D^\perp$, then for any $\boldsymbol{x},\boldsymbol{y} \in R^n$
$$
\boldsymbol{x} + D = \boldsymbol{y} + D \quad \iff \quad 
\left (
\begin{array}{c}
\boldsymbol{h}_1 \cdot \boldsymbol{x}\\
\vdots\\
\boldsymbol{h}_m \cdot \boldsymbol{x}\\ 
\end{array}
\right )
 = 
\left (
\begin{array}{c}
\boldsymbol{h}_1 \cdot \boldsymbol{y}\\
\vdots\\
\boldsymbol{h}_m \cdot \boldsymbol{y}\\ 
\end{array}
\right ).
$$
\end{proposition}
\begin{proof}
Using \eqref{eq:equality_D_Ddualdual}, we have
$$
\begin{array}{rcl}
\boldsymbol{x} +D = \boldsymbol{y} + D	 & \iff &  \boldsymbol{x}-\boldsymbol{y} \in D = D^{\perp \perp}\\
                                         & \iff & \boldsymbol{h} \cdot \boldsymbol{x}  = \boldsymbol{h} \cdot \boldsymbol{y}  \hbox{ for all } \boldsymbol{h} \in D^\perp\\
																				 & \iff & \boldsymbol{h}_i \cdot \boldsymbol{x}  = \boldsymbol{h}_i \cdot \boldsymbol{y}  \hbox{ for } 1 \leq i \leq m. 
\end{array}
$$
\end{proof}

\begin{proof}[Proof of Theorem \ref{theorem:code to parity check system}]
We will have to show that the concatenated matrix constructed in \eqref{eq:matrix representation} satisfy the three conditions given in \eqref{def:parity check system}. That condition (i) is satisfied follows directly from the construction. From Proposition \ref{proposition:cosets} and the fact that $\boldsymbol{d}_i + D \neq \boldsymbol{d}_j + D$ when $i \neq j$, we have that all the columns are distinct in the right part of the concatenated matrix. This shows that condition (ii) is satisfied. It is straightforward to prove that if $\boldsymbol{h} = \sum_{i = 1}^m r_i \boldsymbol{h}_i = \sum_{i = 1}^m r'_i \boldsymbol{h}_i$ for $r_i,r'_i \in R$, then for any coset representative $\boldsymbol{d}_j$
\begin{equation} \label{eq:(iii)_true}
\sum_{i = 1}^m r_i (\boldsymbol{h}_i \cdot \boldsymbol{d}_j) = \sum_{i = 1}^m (r_i \boldsymbol{h}_i) \cdot \boldsymbol{d}_j = \boldsymbol{h} \cdot \boldsymbol{d}_j = \sum_{i = 1}^m (r'_i \boldsymbol{h}_i) \cdot \boldsymbol{d}_j= \sum_{i = 1}^m r'_i (\boldsymbol{h}_i \cdot \boldsymbol{d}_j).
\end{equation}
This implies that condition (iii) also is satisfied.
\end{proof}

\begin{proof}[Proof of Theorem \ref{theorem:parity check system to code}]

For convenience, let $\boldsymbol{h}_i = \hrow{i}$ and set $D = <\boldsymbol{h}_1,\ldots,\boldsymbol{h}_m>^\perp$. We want to define cosets $\boldsymbol{d}_j + D$ by
\begin{equation} \label{eq:wanted_cosets}
\boldsymbol{d}_j + D =
\{ \boldsymbol{x} \in R^n : H \boldsymbol{x}^T = \scol{j} \},
\end{equation}
but we will need to establish that the solution set on the right side is non-empty. Once non-emptiness is established, the rest of the theorem follows easily. Indeed, by Proposition \ref{proposition:cosets}, the cosets are well-defined and disjoint, because the columns of $S$ are distinct. By setting $C = \bigcup_{j = 1}^s (\boldsymbol{d}_j + D)$, it is clear that $C$ has the coset decomposition claimed.

Let $\col(S_D)$ be the set of columns representing all the cosets of $D$ in $R^n$, that is
$$
\col(S_D) = \{ H \boldsymbol{x}^T : \boldsymbol{x} \in R^n\}.
$$
To establish the non-emptiness of the solution set of \eqref{eq:wanted_cosets} it suffices to show that $\col(S) \subseteq \col(S_D)$.

By the definition of parity check systems, every element of $\col(S)$ satisfies (i) and (iii) in \eqref{def:parity check system}. Thus $\col(S) \subseteq S_H$,  where $S_H$ is the set of column vectors that satisfy (i) and (iii) in \eqref{def:parity check system} for the matrix $H$. By inspection and with similar arguments as given in \eqref{eq:(iii)_true}, every element of $\col(S_D)$ satisfies (i) and (iii) in \eqref{def:parity check system} for $H$. Hence, in order to show that $\col(S) \subseteq \col(S_D)$, it suffices to show that $|S_H| \leq |\col(S_D)|$.

We will use induction on $m$ to prove that $|S_H| \leq |\col(S_D)|$. Suppose that $m=1$, then $H$ consists of one row and
$$
\begin{array}{rcl}
S_H & \subseteq & \{ y \in R : y \hbox{ satisfies (i) in \eqref{def:parity check system}} \}\\
    & = & \{\boldsymbol{h}_1 \cdot \boldsymbol{x} : \boldsymbol{x} \in R^n\}\\
		& = & \col(S_D).
\end{array}
$$
Assume that $|S_{H''}| \leq |\col(S_{D''})|$ holds for any $(m-1) \times n$-matrix $H''$ and module $D'' = <H''_{\mathrm{row}(1)},\ldots,H''_{\mathrm{row}(m-1)}>^\perp$, where  $S_{H''}$ and $\col(S_{D''})$ denote the analogue of $S_H$ and $\col(S_D)$.

Now, let $H'$ be the $(m-1) \times n-$matrix with $H'_{\mathrm{row}(i)} = \boldsymbol{h}_i$ and $D' = <\boldsymbol{h}_1,\ldots,\boldsymbol{h}_{m-1}>^\perp$. Hence, by \eqref{eq:equality_D_Ddualdual} and \eqref{eq:cardinality_Ddual_frobenius}, 
\begin{equation} \label{eq:|S_D|}
\begin{array}{rclcl}
|\col(S_D)| & = & |\col(S_{D'})| \cdot \frac{|\col(S_D)|}{|\col(S_{D'})|} & = & |\col(S_{D'}) | \cdot \frac{|R^n| / |D|}{|R^n| / |D'|} \\
            & = & |\col(S_{D'})| \cdot \frac{|D^\perp|}{|D'^\perp|} & = & |\col(S_{D'})| \cdot \frac{|<\boldsymbol{h}_1,\ldots,\boldsymbol{h}_m>|}{|<\boldsymbol{h}_1,\ldots,\boldsymbol{h}_{m-1}>|}.
\end{array}
\end{equation}

Observe, by the definitions of $H'$ and $H$, that 
$$
(x_1,\ldots,x_{m-1}, y)^T \in S_H \Rightarrow (x_1,\ldots,x_{m-1})^T \in S_{H'}.
$$
Choose $\boldsymbol{x}^ T = (x_1,\ldots,x_{m-1})^T \in S_{H'}$ and let
$$
A_{\boldsymbol{x}} = \{y \in R: (x_1,\ldots,x_{m-1}, y)^T \in S_{H} \}.
$$
 By (i) in \eqref{def:parity check system}, $A_{\boldsymbol{x}} \subseteq \{\boldsymbol{h}_m \cdot \boldsymbol{y} : \boldsymbol{y} \in R^n\}$. Let $D_{\cap} = (<\boldsymbol{h}_m> \cap <\boldsymbol{h}_1,\ldots,\boldsymbol{h}_{m-1}>)^\perp$. Condition (iii) in \eqref{def:parity check system} implies that
\begin{equation} \label{eq:in_Dcap}
r y = \sum_{i=1}^{m-1} r r_i x_i
\end{equation}
for $y \in A_{\boldsymbol{x}}$ and $\boldsymbol{z} = r \boldsymbol{h}_m = \sum_{i=1}^{m-1} r_i \boldsymbol{h}_i \in <\boldsymbol{h}_m> \cap <\boldsymbol{h}_1,\ldots,\boldsymbol{h}_{m-1}>$.  Consequently, by Proposition \ref{proposition:cosets}, 
$$
A_{\boldsymbol{x}} \subseteq \{\boldsymbol{h}_m \cdot \boldsymbol{y} : \boldsymbol{y} \in B_{\boldsymbol{x}}\},
$$
where $B_{\boldsymbol{x}}$ is the coset of $D_{\cap}$ such that \eqref{eq:in_Dcap} is satisfied. 

Let $D_m = <\boldsymbol{h_m}>^\perp$. We observe that $D_m$ is a submodule of $D_\cap$. Hence $B_{\boldsymbol{x}}$ equals a union of cosets of $D_m$. Therefore, as there is a one-to-one correspondence between the cosets of $D_m$ and the elements in $\{\boldsymbol{h}_m \cdot \boldsymbol{y} : \boldsymbol{y} \in R^n\}$ by Proposition \ref{proposition:cosets},
$$
|\{\boldsymbol{h}_m \cdot \boldsymbol{y} : \boldsymbol{y} \in B_{\boldsymbol{x}}\}| = \frac{|D_\cap|}{|D_m|}.
$$ 
Consequently, by \eqref{eq:cardinality_Ddual_frobenius}, 
$$
|A_{\boldsymbol{x}}| \leq \frac{|D_\cap|}{|D_m|} = \frac{|R|^n / |D_\cap^\perp|}{|R|^n / |D_m^\perp|} = \frac{|<\boldsymbol{h}_m>|}{|(<\boldsymbol{h}_m> \cap <\boldsymbol{h}_1,\ldots,\boldsymbol{h}_{m-1}>)|}.
$$
This implies that,
\begin{equation} \label{eq:|S_H|}
|S_H| = \sum_{\boldsymbol{x} \in S_{H'}} |A_{\boldsymbol{x}}| \leq |S_{H'}| \cdot \frac{|<\boldsymbol{h}_m>|}{|(<\boldsymbol{h}_m> \cap <\boldsymbol{h}_1,\ldots,\boldsymbol{h}_{m-1}>)|}
\end{equation}

By a well-known isomorphism theorem for modules, see for example Theorem 6.38 in \cite{sharp00}, the factor modules
$$
<\boldsymbol{h}_m> / (<\boldsymbol{h}_m> \cap <\boldsymbol{h}_1,\ldots,\boldsymbol{h}_{m-1}>) \hbox{ and } <\boldsymbol{h}_1,\ldots,\boldsymbol{h}_m> / <\boldsymbol{h}_1,\ldots,\boldsymbol{h}_{m-1}>
$$
are isomorphic. Consequently, by \eqref{eq:|S_D|}, \eqref{eq:|S_H|} and the assumption that $|S_{H'}| \leq |\col(S_{D'})|$,
$$
\begin{array}{rcl}
|S_H| & \leq & |S_{H'}| \cdot \frac{|<\boldsymbol{h}_1,\ldots,\boldsymbol{h}_m>|}{|<\boldsymbol{h}_1,\ldots,\boldsymbol{h}_{m-1}>|}\\
      & \leq & |\col(S_{D'})| \cdot \frac{|<\boldsymbol{h}_1,\ldots,\boldsymbol{h}_m>|}{|<\boldsymbol{h}_1,\ldots,\boldsymbol{h}_{m-1}>|}\\
			& = & |\col(S_D)|.
\end{array} 
$$
\end{proof}

Let $\decomp$ be the code that we get from a parity check system $(H|S)$ in Theorem \ref{theorem:parity check system to code}. We obtain an $R$-module structure on 
$$
\col(S_{D_C}) = \{H \boldsymbol{z}^T : \boldsymbol{z} \in R^n\}
$$
by defining
$$
\boldsymbol{x}^T + \boldsymbol{y}^T = (x_1 + y_1,\ldots,x_m + y_m)^T
\quad \hbox{and} \quad
r \boldsymbol{x}^T = (rx_1,\ldots,rx_m)^T, 
$$ 
for $\boldsymbol{x}^T, \boldsymbol{y}^T \in \col(S_{D_C})$ and $r \in R$. Further, let 
$$
\mathrm{ker}(\mathrm{col}(S)) = \{\boldsymbol{x}^T \in \col(S_{D_C}): r \boldsymbol{x}^T + \col(S) = \col(S) \hbox{ for all } r\in R\}.
$$
Since $H \boldsymbol{x}^T = (0,\ldots,0)^T$ if and only if $\boldsymbol{x} \in D_C$, we obtain that $\col(S_{D_C})$  is isomorphic to the factor module $R^n/D_C$ and $\mathrm{ker}(C)/D_C$ is a submodule of $R^n/D_C$. It is now straightforward to show the following proposition.

\begin{proposition}
Let $\decomp$ be the code that we get from a parity check system $(H|S)$ in Theorem \ref{theorem:parity check system to code}. Then
\begin{align*}
\hbox{(i) } & \mathrm{ker}(C) = \{ \boldsymbol{x} \in R^n : H \boldsymbol{x}^T \in \mathrm{ker}(\mathrm{col}(S)) \},\\ 
\hbox{(ii) } & C \hbox{ is a linear code} \quad \iff \quad \mathrm{col}(S) \hbox{ is a submodule of } \col(S_{D_C}).
\end{align*}
\end{proposition}

We remark that a parity check matrix $H$ to a linear code $C$ over a finite field $\mathbb{F}_q$ corresponds to the parity check system $(H|\boldsymbol{0}^T)$. Moreover, it is straightforward to show that if the rows in a concatenation $(H|S)$ over $\mathbb{F}_q$ is linear independent, then $(H|S)$ is a parity check system if and only if condition (ii) in \eqref{def:parity check system} is satisfied.  Further, we observe that if $<\hrow{1},\ldots,\hrow{m}> \cong \oplus_{i = 1}^m  <\hrow{i}>$ for a concatenation $(H|S)$ over $R$, then (iii) in \eqref{def:parity check system} is trivially satisfied and the conditions (i) and (ii) in \eqref{def:parity check system} is enough to define a parity check system. 

As mentioned before, parity check systems for nonlinear codes over $\mathbb{F}_2$ was defined in \cite{heden08a}. In \cite{heden08a}, when defining parity check systems $(H|S)$ over $\mathbb{F}_2$, only condition (ii) in \eqref{def:parity check system} and the condition that the rows in $H$ are linear independent are used. These conditions are enough to use when defining parity check systems over finite fields, but not when defining parity check systems in general over finite commutative Frobenius rings. Moreover, in \cite{heden08a} it is always assumed that the zero-word $\boldsymbol{0}$ is contained in the binary nonlinear codes. This corresponds to the property that the zero-column $\boldsymbol{0}^T$ is contained in $\col(S)$. Hence, the zero-column is omitted in $\col(S)$ in the parity check systems given in \cite{heden08a}. (In the present paper we do not assume that the zero-word is contained in the codes. Therefore we do not omit the zero-column in $\col{S}$ when the zero-word is contained in the code.)   



\section{Fourier analysis on codes over finite commutative Frobenius rings} \label{section:dfa}

In Section \ref{subsection:dfa and finite abelian groups} we give some standard results on Fourier analysis on finite Abelian groups and introduce some notation. For an introduction to Fourier analysis on finite Abelian groups and for proofs of the results in Section \ref{subsection:dfa and finite abelian groups}, see for example Section 10 and 12 in \cite{terras99}. In Section \ref{subsec:DFA on R}, we give some results about Fourier analysis on $R^n$. Most of these results, except for Theorem  \ref{theorem:Fourier code} and \ref{theorem:poisson}, can be found in \cite{wood99}.

\subsection{Fourier analysis on finite Abelian groups $G$} \label{subsection:dfa and finite abelian groups}

Let $G$ be a finite (additive) Abelian group. By $\mathbb{C}^G$ we denote the vector space over $\mathbb{C}$ consisting of all functions from $G$ to $\mathbb{C}$. The addition of vectors and multiplication with scalars are defined in the following way: for $f,h:G \rightarrow \mathbb{C}$, $c \in \mathbb{C}$ and $x \in G$,
$$
\begin{array}{l}
(f + h)(x) = f(x) + h(x), \\
(cf)(x) = cf(x).
\end{array}
$$
Let $\langle \hbox{ , } \rangle$ denote the Hermitian inner product on $\mathbb{C}^G$ defined by
$$
\langle f,h \rangle = \sum_{x \in G} f(x) \overline{h(x)},
$$
where  $\overline{h(x)}$ denotes the complex conjugate of $h(x)$. The \emph{convolution} $f*h$, is defined by
$$
(f * h)(x) = \sum_{y \in G} f(y) h(x-y),
$$
and gives the vector space $\mathbb{C}^G$ the structure of an associative commutative algebra.

For $x,y \in G$, let $\delta_x$ denote the \emph{indicator function}, defined by
$$
\delta_x(y) = 
\left \{
\begin{array}{lll}
1 & \hbox{if} & x = y,\\
0 & \hbox{if} & x \neq y.
\end{array}
\right.
$$
The set of functions $\{\delta_{x} | x \in G\}$, constitutes an orthonormal basis of $\mathbb{C}^G$. 

Let $\mathbb{C}^\times$ denote the set $\mathbb{C} \setminus \{0\}$. A \emph{character} of a finite Abelian group $G$ is a group homomorphism from $G$ into the multiplicative group of $\mathbb{C^\times}$. The set of all characters constitutes a group $\widehat{G}$ under pointwise multiplication of the characters, that is,
\begin{equation} \label{eq:character_mult}
(\psi \chi) (x) = \psi(x) \chi(x), 
\end{equation}
for any $\psi, \chi \in \widehat{G}$ and $x \in G$. 

The set of characters of a finite Abelian group can explicitly be described as follows. Suppose that $G$ is equal to the finite Abelian group $\mathbb{Z}_{t_1} \oplus \ldots \oplus \mathbb{Z}_{t_k}$. For any $x = (x_1,\ldots,x_k), y = (y_1,\ldots,y_k) \in G$, define
\begin{equation} \label{eq:character abelian}
e_x(y) = \prod_{j = 1}^k e^{\frac{2 \pi i}{t_j} x_j y_j}.
\end{equation}
The set of maps $\{e_x | x \in G\}$ is the set of characters of $G$.

For any $\chi \in \widehat{G}$ and $x \in G$, we have that 
\begin{equation} \label{eq:minus_char}
\chi(-x) = \overline{\chi(x)} \quad \hbox{and} \quad \chi(0_G) = 1.
\end{equation}
If $G$ is a direct sum $G = G_1 \oplus G_2$, then all characters of $G$ has the form $\chi = (\chi_1,\chi_2) \in \widehat{G_1} \oplus \widehat{G_2}$ where
\begin{equation} \label{eq:direct sum of character}
\chi(g) = \chi_1(g_1) \chi_2(g_2)
\end{equation} 
for $g = (g_1,g_2)$ in $G$. Consequently, $\widehat{G}$ is isomorphic to $\widehat{G_1} \oplus \widehat{G_2}$. 

The character group $\widehat{G}$ is isomorphic to the finite Abelian group $G$. Under a fixed group isomorphism of $G$ with $\widehat{G}$, write $\chi_x$ for the image of every $x \in G$. The set of characters constitutes an orthogonal basis of $\mathbb{C}^G$, where
$$
\langle \chi_x , \chi_x \rangle = | G |\quad \hbox{and} \quad
\chi_x * \chi_y =
\left \{
\begin{array}{lcl}
| G | \chi_x & \hbox{if} & x = y,\\
0 & \hbox{if} & x \neq y.
\end{array}
\right.
$$
Any function $f:G \rightarrow \mathbb{C}$ is represented in $\mathbb{C}^G$ by the bases $\{\delta_x : x \in G\}$ and $\{\chi_x : x \in G\}$ as
$$
\sum_{x \in G} f(x) \delta_x \quad \hbox{ and } \quad  \frac{1}{|G|}\sum_{x \in G} \hat{f}(x) \chi_x,
$$  
where 
\begin{equation} \label{eq:fourier transform}
f(x) = \frac{1}{| G |} \sum_{y \in G} \hat{f}(y) \chi_y(x) \quad \hbox{and} \quad \hat{f}(x) = \sum_{y \in G} f(y) \chi_x(-y).
\end{equation}
The coefficients $\hat{f}(x)$ are called the \emph{Fourier coefficients} of the function $f$. A key property for many results using Fourier coefficients is that   
$$
\widehat{(f * h)}(x) =  \hat{f}(x) \hat{h}(x)
$$
for any $f, h \in \mathbb{C}^G$ and $x \in G$.

For any subset $B \subseteq G$, let
\begin{equation} \label{eq:def_character_dual}
B^\lozenge = \{\chi_x \in \widehat{G} | \chi_x(b) = 1 \hbox{ for all } b \in B\}.
\end{equation}
Let $D$ be a subgroup of $G$. Then the \emph{character dual} $D^\lozenge$ of $D$ is a subgroup of $\widehat{G}$ isomorphic to $\widehat{G/D}$. Therefore, 
\begin{equation} \label{eq:mid Dlozenge mid} 
| D^{\lozenge} | \hbox{ = } | G | /| D |. 
\end{equation}
Moreover,
\begin{equation}\label{eq:Fourier orthogonality}
\sum_{d \in D} \chi_x(d) =
\left \{
\begin{array}{ccl}
| D | & \hbox{if} & \chi_x \in D^\lozenge,\\
0 & \hbox{if} & \chi_x \notin D^\lozenge.
\end{array}
\right .
\end{equation}
For any element $y \in G$ and function $f \in \mathbb{C}^G$, the Poisson summation formula is as follows,
\begin{equation} \label{eq:coset poisson}
\sum_{d \in D} f(y + d) = \frac{| D |}{| G |} \sum_{\chi_x \in D^\lozenge} \hat{f}(x)\chi_x(y).
\end{equation}


\subsection{Fourier analysis on $R$-modules $R^n$} \label{subsec:DFA on R}

Let $\widehat{R}$ denote the character group of the additive group of $R$.  From \eqref{eq:direct sum of character}, $\widehat{R^n}$ and ${\widehat{R}}^n$ are isomorphic as groups and 
\begin{equation} \label{eq:chi_x}
\boldsymbol{\chi_x}(\boldsymbol{y}) = \prod_{i=1}^n \chi_{x_i}(y_i)
\end{equation}
for $\boldsymbol{\chi_x} = (\chi_{x_1},\ldots,\chi_{x_n}) \in \widehat{R}^n$ and $\boldsymbol{y} = (y_1,\ldots,y_n) \in R^n$. Consequently, for any subset $B \subseteq R^n$,
$$
B^{\lozenge} = \{\boldsymbol{\chi_x} \in \widehat{R}^n : \prod_{i = 1}^n \chi_{x_i}(y_i) = 1 \hbox{ for all } \boldsymbol{y} \in B\}
$$ 
We have an $R$-module structure on $\widehat{R}$ by the operation
$$
\chi_x^r(y) = \chi_x(ry)
$$
for $\chi_x \in \widehat{R}$ and  $r,y \in R$. Thus, we also have an $R$-module structure on $\widehat{R}^n$ by the operation
$$
\boldsymbol{\chi_x}^r = (\chi_{x_1}^r,\ldots,\chi_{x_n}^r)
$$
for $\boldsymbol{\chi_x} \in \widehat{R}^n$ and  $r \in R$. 

A \emph{generating character} of $\widehat{R}$ is a character $\chi$ such that $\widehat{R} = \{\chi^r : r \in R\}$. Finite commutative Frobenius rings can be characterized by generating characters. Namely, a finite commutative ring $R$ is Frobenius if and only if $\widehat{R}$ has a generating character. Examples of generating characters for some finite commutative Frobenius rings $R$ are: 
$$
\chi(y) = e^{\frac{2 \pi i}{t} y} \hbox{ for } y \in R = \mathbb{Z}_t \hbox{ and } \chi(y) = e^{\frac{2 \pi i}{p} \mathrm{Tr}(y)} \hbox{ for } y \in R = \mathbb{F}_q,
$$ 
where $q = p^l$ for some prime $p$ and $\mathrm{Tr}(y) = y + y^p + \ldots + y^{p^{l-1}}$. See \cite{wood99} for some more examples of generating characters.

Let $R$ be finite commutative Frobenius ring and $\varepsilon$ be a generating character of $\widehat{R}$. Then $R$ and $\widehat{R}$ are isomorphic as $R$-modules via the map $\phi : R \rightarrow \widehat{R}$, where $\phi(x) = \varepsilon^x$. Further, let
\begin{equation} \label{eq:def_generating_char}
\echar{x}= (\varepsilon^{x_1},\ldots,\varepsilon^{x_n})
\end{equation}
for $\boldsymbol{x} \in R^n$. Then $R^n$ and $\widehat{R}^n$ are isomorphic as $R$-modules via the map $\varphi: R^n \rightarrow \widehat{R}^n$, where
\begin{equation} \label{eq:module_iso}
\varphi(\boldsymbol{x}) = \echar{x}.
\end{equation}
Hence, by \eqref{eq:fourier transform}, any function $f:R^n \rightarrow \mathbb{C}$ is represented in $\mathbb{C}^{R^n}$ by the bases $\{\delta_{\boldsymbol{x}} : \boldsymbol{x} \in R^n\}$ and $\{\echar{x} : \boldsymbol{x} \in R^n\}$ as
$$
\sum_{\boldsymbol{x} \in R^n} f(\boldsymbol{x}) \delta_{\boldsymbol{x}} \quad \hbox{ and } \quad  \frac{1}{|R|^n}\sum_{\boldsymbol{x} \in R^n} \hat{f}(\boldsymbol{x}) \echar{x},
$$  
where 
\begin{equation} \label{eq:fourier transform M}
f(\boldsymbol{x}) = \frac{1}{| R |^n} \sum_{\boldsymbol{y} \in R^n} \hat{f}(\boldsymbol{y}) \echararg{y}{x} \quad \hbox{and} \quad \hat{f}(\boldsymbol{x}) = \sum_{\boldsymbol{y} \in R^n} f(\boldsymbol{y}) \echarargminus{x}{y}.
\end{equation}

By \eqref{eq:chi_x} and the property that $\varepsilon^{r + r'} = \varepsilon^{r} \varepsilon^{r'} $ for $r,r' \in R$, it follows that
\begin{equation} \label{eq:char_value}
\echararg{x}{y} = \prod_{i = 1}^n \varepsilon^{x_i}(y_i) = \prod_{i = 1}^n \varepsilon^{x_i y_i}(1) = \varepsilon^{\boldsymbol{x} \cdot \boldsymbol{y}}(1) = \varepsilon (\boldsymbol{x} \cdot \boldsymbol{y})
\end{equation}
for $\boldsymbol{x}, \boldsymbol{y} \in R^n$. Now, let $D$ be a submodule of $R^n$. Then, by \eqref{eq:cardinality_Ddual_frobenius} and\eqref{eq:mid Dlozenge mid},
$$
|D^\lozenge| = |D^\perp| = \frac{|R|^n}{|D|}.
$$
For any $\boldsymbol{x} \in D^\perp$, we have $\echar{x} \in D^\lozenge$ by \eqref{eq:char_value}. Hence, 
\begin{equation} \label{eq:lozenge_perp}
D^\lozenge = \{\echar{x} : \boldsymbol{x} \in D^\perp\}.
\end{equation}

Any code $C$ of $R^n$ can be represented by an element $\delta_C$ of $\mathbb{C}^{R^n}$, where 
$$
\delta_C(\boldsymbol{x}) = 
\left \{
\begin{array}{lcl}
1 & \hbox{if} & \boldsymbol{x} \in C,\\
0 & \hbox{if} & \boldsymbol{x} \notin C.
\end{array}
\right .
$$
Consequently, 
\begin{equation} \label{eq:def_fourier_repr_C}
\delta_{C} = \sum_{\boldsymbol{c} \in C} \delta_{\boldsymbol{c}}.
\end{equation}

\begin{theorem} \label{theorem:Fourier code}
For any code $C = \bigcup_{j = 1}^s (\boldsymbol{d}_j + D_C)$,
$$
\hat{\delta}_C(\boldsymbol{x}) =
\left \{
\begin{array}{ccl}
| D_C | \sum_{j=1}^s \echarargminusindex{x}{d}{j} & \hbox{if} & \boldsymbol{x} \in D_C^\perp,\\
0	& \hbox{if} & \boldsymbol{x} \notin D_C^\perp.
\end{array}
\right .
$$
\end{theorem}
\begin{proof}
By \eqref{eq:fourier transform M},
\begin{align*}
\hat{\delta}_C(\boldsymbol{x})  & = 
\sum_{\boldsymbol{y} \in R^n} \delta_C (\boldsymbol{y}) \echarargminus{x}{y}   =
\sum_{\boldsymbol{c} \in C} \echarargminus{x}{c} = 
\sum_{i = 1}^s \sum_{\boldsymbol{d} \in D_C} \echarargminuscoset{x}{d}{i}\\
& = \sum_{i = 1}^s \sum_{\boldsymbol{d} \in D_C} \echarargminusindex{x}{d}{i} \echararg{x}{d} = 
\sum_{i=1}^s \echarargminusindex{x}{d}{i} \left ( \sum_{\boldsymbol{d} \in D_C} \echararg{x}{d} \right ). 
\end{align*}
The theorem now follows using \eqref{eq:Fourier orthogonality} and \eqref{eq:lozenge_perp}.
\end{proof}

\begin{example}
Let $\decomparg{3}$ be the code over $\mathbb{Z}_6$ in Example \ref{ex:C_to_PC}, where
$$
\begin{array}{l}
\boldsymbol{d}_1 = (0,0,0,0) \hbox{, } \boldsymbol{d}_2 = (5,2,0,0) \hbox{, } \boldsymbol{d}_3 = (4,1,0,0),\\
D_C = <(2,1,1,0),(0,1,0,1),(3,0,3,0)>,\\
D_C^\perp = <(1,1,3,5),(0,4,2,2)>.
\end{array}
$$
First we observe that 
$$
|D_C| = \frac{|\mathbb{Z}_6^4|}{|D_C^\perp|} = \frac{6^4}{18} = 72.
$$
Choose $\varepsilon(x) = e^{\frac{2 \pi i}{6} x}$ as a generating character for $\widehat{\mathbb{Z}_6}$. Let $\boldsymbol{x}$ be the element $1(1,1,3,5) + 2(0,4,2,2) = (1,3,1,3) \in D_C^\perp$. Then, by Theorem \ref{theorem:Fourier code} and \eqref{eq:char_value},
$$
\begin{array}{rcl}
\hat{\delta}_C(\boldsymbol{x}) & = & | D_C | \sum_{j=1}^3 \echarargminusindex{x}{d}{j}\\
                               & = & 72 (\boldsymbol{\varepsilon}_{(1,3,1,3)}(0,0,0,0) + \boldsymbol{\varepsilon}_{(1,3,1,3)}(1,4,0,0) + \boldsymbol{\varepsilon}_{(1,3,1,3)}(2,5,0,0))\\
															 & = & 72(\varepsilon ((1,3,1,3) \cdot (0,0,0,0)) + \varepsilon ((1,3,1,3) \cdot (1,4,0,0)) + \\
															 && \varepsilon ((1,3,1,3) \cdot (2,5,0,0)))\\
															 & = & 72(\varepsilon(0) + \varepsilon(1) + \varepsilon(5))\\
															 & = & 72(e^{\frac{2 \pi i}{6} 0} + e^{\frac{2 \pi i}{6} 1} + e^{\frac{2 \pi i}{6} 5} )\\
															 & = & 144.
\end{array}      
$$
\end{example}

For any code $C$, let $-C = \{- \boldsymbol{c} : \boldsymbol{c} \in C\}$.

\begin{theorem} \label{theorem:poisson}
For any code $\decomp$ and $f \in \mathbb{C}^{R^n}$,
$$
\sum_{\boldsymbol{c} \in C} f(\boldsymbol{c}) = 
\frac{1}{| R | ^n} \sum_{\boldsymbol{x} \in D_C^\perp} \hat{f}(\boldsymbol{x}) \hat{\delta}_{-C}(\boldsymbol{x}).
$$
\end{theorem}
\begin{proof}
By \eqref{eq:coset poisson}, \eqref{eq:fourier transform M} and \eqref{eq:lozenge_perp} 
\begin{align*}	
\sum_{\boldsymbol{c} \in C} f(\boldsymbol{c}) & 
= \sum_{j = 1}^s  \sum_{\boldsymbol{d} \in D_C} f(\boldsymbol{d}_j + \boldsymbol{d}) 
= \sum_{j = 1}^s \frac{| D_C |}{| R |^n} \sum_{\boldsymbol{x} \in D_C^\perp} \hat{f}(\boldsymbol{x}) \echarargindex{x}{d}{j} \\
& = \frac{1}{| R |^n} \sum_{\boldsymbol{x} \in D_C^\perp} \hat{f}(\boldsymbol{x}) \left (| D_C | \sum_{j=1}^s \echarargindex{x}{d}{j} \right ).
\end{align*}
Since $-C$ has the coset decomposition $\bigcup_{j = 1}^s (-\boldsymbol{d}_j + D_C)$, the theorem now follows from Theorem \ref{theorem:Fourier code}. 
\end{proof}


\section{Parity check systems, distance and Fourier coefficients} \label{sec:distance}

Section \ref{subsec:distance and error} deals with how to get the minimum distance of a code and how to do error-correcting by the use of a parity check system.  In general, it is not easy to give values to the Fourier coefficients $\hat{f}(\boldsymbol{x})$ of a member $f$ in $\mathbb{C}^{R^n}$ such that $f$ equals $\delta_C$ for some code $C$ over $R$. In section \ref{subsec:fourier coefficient} we give a formula on how to get the Fourier coefficients $\widehat{\delta}_C(\boldsymbol{x})$ of a code $C$ by the use of an associated parity check system. A MacWiliam type of distance identity between the code and an associated parity check system is given in Section \ref{subsec: identity distance}.

\subsection{Minimum distance and error-correcting} \label{subsec:distance and error}

As defined in \eqref{eq:min_supp_weight} and \eqref{eq:mindist}, we recall that $\wt{x}$ and $\mindist{C}$ denote the Hamming weight and distance of an element $\boldsymbol{x} \in R^n$ and a code $C$, respectively. For any parity check system $(H | S)$, let $\sdist$ be the following set of column vectors of size $m$,
$$
\sdist = \{\scol{l} - \scol{k} : 1 \leq k < l \leq s\} \cup (0,\ldots,0)^T
$$
The following theorem shows how a parity check system of a code $C$ can be used to derive the minimum distance of $C$.

\begin{theorem} \label{theorem:d(C)} 
Let $(H|S)$ be a parity check system of a code $C$. Then
$$
\mindist{C} = \min \{ \wt{x} : \boldsymbol{x} \in R^n \setminus \{\boldsymbol{0}\} \hbox{ and }  H \boldsymbol{x}^T \in  \sdist\}.
$$
\end{theorem}

\begin{proof}
Let $\bigcup_{j = 1}^s (\boldsymbol{d}_j + D_C)$ be the coset decomposition of $C$ that corresponds to the parity check system $(H|S)$ in Theorem \ref{theorem:parity check system to code}. The minimum distance of $C$ equals the least positive integer $d$ such that 
$$
\boldsymbol{x} + \boldsymbol{d}_{k} + D_C = \boldsymbol{d}_{l} + D_C
$$
for an element $\boldsymbol{x} \in R^n \setminus \{\boldsymbol{0}\}$ with $\wt{x} = d$ and some $k,l \in [s]$. By Proposition \ref{proposition:cosets} and the property that $D_C = <\hrow{1},\ldots,\hrow{m}>^\perp$ we may conclude that the equality above is equivalent to 
$$
H (\boldsymbol{x} + \boldsymbol{d}_{k})^T = H \boldsymbol{d}_{l}^T \iff H \boldsymbol{x}^T = H \boldsymbol{d}_{l}^T - H \boldsymbol{d}_{k}^T \iff H (-\boldsymbol{x})^T = H \boldsymbol{d}_{k}^T - H \boldsymbol{d}_{l}^T.
$$
The theorem now follows from the observations that $\wt{x} = \wtminus{x}$ and that $k = l$ implies that $H \boldsymbol{x}^T = \boldsymbol{0}^T$ in the equality above. 
\end{proof}

\begin{example}
Let $C$ be the code associated to the following parity check system $(H|S)$ over $\mathbb{Z}_6$,
$$
(H|S) = 
\left (
\begin{array}{cccc|ccc}
1&1&3&5&0&1&5\\
0&4&2&2&0&2&4
\end{array}
\right )
.
$$
This gives that $\sdist = \{(0,0)^T,(1,2)^T, (5,4)^T, (4,2)^T\}$ and that the minimum distance of $C$ equals to 2 since
$$
\{H \boldsymbol{x}^T : \boldsymbol{x} \in \mathbb{Z}_6^4, \wt{x} = 1\} \cap \sdist = \emptyset
$$
and $H (5,0,0,1)^T = (4,2)^T \in \sdist$.
\end{example}

We recall that $\dist{x}{y}$ denotes the distance between two elements $\boldsymbol{x}, \boldsymbol{y} \in R^n$, as defined in \eqref{eq:min_supp_weight}. For any code $C$ and $\boldsymbol{x} \in R^n$, let $\cdist{x} = \min \{\dist{x}{c} : \boldsymbol{c} \in C\}$. It is straightforward to show that if $\cdist{x} \leq \lfloor \frac{\mindist{C} - 1}{2} \rfloor$, then there is a unique codeword $\boldsymbol{c} \in C$, here called the \emph{nearest neighbor} to $\boldsymbol{x}$ in $C$, such that $\dist{x}{c} = \cdist{x}$.

\begin{theorem} \label{theorem:error-correct} 
Let $(H|S)$ be a parity check system of a code $C$ and let $\boldsymbol{x} \in R^n$ be an element such that $\cdist{x} \leq \lfloor \frac{\mindist{C} - 1}{2} \rfloor\}$. Then there is a unique $l \in [s]$ and $\boldsymbol{y} \in R^n$ with $\wt{y} = \cdist{x}$ such that
$$
H \boldsymbol{x}^T = \scol{l} + H \boldsymbol{y}^T.
$$
The nearest neighbor to $\boldsymbol{x}$ in $C$ is the codeword $\boldsymbol{x} - \boldsymbol{y}$.
\end{theorem}
\begin{proof}
Let $\bigcup_{j = 1}^s (\boldsymbol{d_j} + D_C)$ be the coset decomposition to $C$ that corresponds to the parity check system $(H|S)$ in Theorem \ref{theorem:parity check system to code}. Let $\boldsymbol{c}$ be the unique nearest neighbor to $\boldsymbol{x}$ in $C$ and let $\boldsymbol{y} = \boldsymbol{x} - \boldsymbol{c}$. Further, let $l$ be the element in $[s]$ such that $\boldsymbol{c} \in \boldsymbol{d}_{l} + D_C$. Then $\wt{y} = \cdist{x}$, and by Proposition \ref{proposition:cosets} we obtain that
$$
(\boldsymbol{x} - \boldsymbol{y}) + D_C = (\boldsymbol{d}_{l} + D_C)
\quad \iff \quad 
H \boldsymbol{x}^T = \scol{l} + H \boldsymbol{y}^T.
$$
\end{proof}

\subsection{The Fourier coefficient $\hat{\delta}_C (\boldsymbol{x})$} \label{subsec:fourier coefficient}

For any parity check system $(H|S)$  and linear combination $\boldsymbol{x} = \sum_{i = 1}^m r_i \hrow{i}$, let $\srowelem{x}$ denote the following vector in $R^s$:
\begin{equation} \label{eq:def_Sx}
\srowelem{x} = (\srowelem{x}(1),\ldots,\srowelem{x}(s)) = \sum_{i = 1}^m r_i \srow{i}.
\end{equation}
From (iii) in \eqref{def:parity check system}, we have that $\srowelem{x}$ is well-defined. 

The next theorem shows how the Fourier coefficients of a code $C$ can be derived from a parity check system of $C$. We recall that $\varepsilon$ is a generating character of $\widehat{R}$. Moreover, see \eqref{eq:def_generating_char} for the definition of $\echar{x}$.

\begin{theorem} \label{theorem:Fourier parity check system}
Let $(H|S)$ be a parity check system of a code $C$. Then, for $\boldsymbol{x} \in R^n$, 
$$
\hat{\delta}_C(\boldsymbol{x}) =
\left \{
\begin{array}{ccl}
\frac{| R |^n}{|<\hrow{1},\ldots,\hrow{m}>|} 
\sum_{j=1}^s \varepsilon(-\srowelem{x}(j)) & \hbox{if} & \boldsymbol{x} \in <\hrow{1},\ldots,\hrow{m}>,\\
0	& \hbox{if} & \boldsymbol{x} \notin <\hrow{1},\ldots,\hrow{m}>.
\end{array}
\right .
$$
\end{theorem}
\begin{proof}
By Theorem \ref{theorem:parity check system to code}, there is a coset decomposition $\bigcup_{j = 1}^s (\boldsymbol{d}_j + D_C)$ of $C$ such that $D_C = <\hrow{1},\ldots,\hrow{m}>^\perp$ and $s_{i,j} = \hrow{i} \cdot \boldsymbol{d}_j$. Hence, for $\boldsymbol{x} = \sum_{i = 1}^m r_i \hrow{i}$,
$$
\srowelem{x} = (\sum_{i = 1}^m r_i (\hrow{i} \cdot \boldsymbol{d}_1), \ldots, \sum_{i = 1}^m r_i (\hrow{i} \cdot \boldsymbol{d}_s)) = (\boldsymbol{x} \cdot \boldsymbol{d}_1, \ldots, \boldsymbol{x} \cdot \boldsymbol{d}_s).
$$
From \eqref{eq:char_value} we have that
$$
\echarargindex{x}{d}{j} =  \varepsilon(\srowelem{x}(j)).
$$
By  \eqref{eq:equality_D_Ddualdual} and\eqref{eq:cardinality_Ddual_frobenius}, we deduce that
$$
| D_C | = \frac{| R |^n}{|<\hrow{1},\ldots,\hrow{m}>|}.
$$
The theorem now follows from Theorem \ref{theorem:Fourier code}.
\end{proof}

\begin{example} \label{ex:fourier_coeff}
Let $C$ be the code associated to the parity check system
$$
(H|S) = 
\left (
\begin{array}{cccc|ccc}
1&1&3&5&0&1&5\\
0&4&2&2&0&2&4
\end{array}
\right )
$$
over the ring $R = \mathbb{Z}_6$. Choose $\varepsilon(x) = e^{\frac{2 \pi i}{6} x}$ as a generating character for $\widehat{\mathbb{Z}_6}$. Observe that
$$
\frac{|R|^n}{|<\hrow{1}, \ldots, \hrow{m} >|} = \frac{|\mathbb{Z}_6|^4}{|<(1,1,3,5),(0,4,2,2)>|} = \frac{6^4}{18}= 72.
$$
By Theorem \ref{theorem:Fourier parity check system} we obtain the following list of values for $\hat{\delta}_{C}(\boldsymbol{x})$.
$$
\begin{array}{|c|c|c|c|}
\hline
\boldsymbol{x} \in <\hrow{1}, \hrow{2}> & \srowelem{x} & \sum_{j=1}^3 \varepsilon(-\srowelem{x}(j)) & \hat{\delta}_C(\boldsymbol{x})\\
\hline
\hline
(0,0,0,0) & (0,0,0) & 3 & 216\\
\hline
(1,1,3,5) & (0,1,5) & 2 & 144\\
\hline
(2,2,0,4) & (0,2,4) & 0 & 0\\
\hline
(3,3,3,3) & (0,3,3) & -1 & -72\\
\hline
(4,4,0,2) & (0,4,2) & 0 & 0\\
\hline
(5,5,3,1) & (0,5,1) & 2 & 144\\
\hline
(0,4,2,2) & (0,2,4) & 0 & 0\\
\hline
(1,5,5,1) & (0,3,3) & -1 & -72\\
\hline
(2,0,2,0) & (0,4,2) & 0 & 0\\
\hline
(3,1,5,5) & (0,5,1) & 2 & 144\\
\hline
(4,2,2,4) & (0,0,0) & 3 & 216\\
\hline
(5,3,5,3) & (0,1,5) & 2 & 144\\
\hline
(0,2,4,4) & (0,4,2) & 0 & 0\\
\hline
(1,3,1,3) & (0,5,1) & 2 & 144\\
\hline
(2,4,4,2) & (0,0,0) & 3 & 216\\
\hline
(3,5,1,1) & (0,1,5) & 2 & 144\\
\hline
(4,0,4,0) & (0,2,4) & 0 & 0\\
\hline
(5,1,1,5) & (0,3,3) & -1 & -72\\
\hline
\end{array}
$$
\end{example}

\subsection{An identity for the distance} \label{subsec: identity distance}

Let 
$$
X^2 = \{ (\boldsymbol{x_1}, \boldsymbol{x_2}) : \boldsymbol{x_1}, \boldsymbol{x_2} \in X\}
$$ 
for any subset $X \subseteq R^n$. If $D$ is a submodule of $R^n$, then we immediately have that $D^2$ is a submodule of $(R^n)^2 = R^{2n}$. By inspection, 
$$
\bigcup_{(k,l) \in [s] \times [s]} \left( (\boldsymbol{d_k},\boldsymbol{d_l}) + D_C^2 \right)
$$
is a coset decomposition of $C^2$ for a code $\decomp$. The dot product of $(R^n)^2$ is defined by 
$$
(\boldsymbol{x_1}, \boldsymbol{x_2}) \cdot (\boldsymbol{x'_1}, \boldsymbol{x'_2}) = \boldsymbol{x_1} \cdot \boldsymbol{x'_1} +  \boldsymbol{x_2} \cdot \boldsymbol{x'_2}
$$
for $(\boldsymbol{x_1}, \boldsymbol{x_2}), (\boldsymbol{x'_1}, \boldsymbol{x'_2}) \in (R^n)^2$, and we observe that $(D_C^2)^\perp = (D_C^\perp)^2$. By (\ref{eq:direct sum of character}) and some straightforward calculations, we note that  Theorem \ref{theorem:Fourier code} implies that 
$$
\hat{\delta}_{C^2}(\boldsymbol{x_1},\boldsymbol{x_2}) = \hat{\delta}_C(\boldsymbol{x_1}) \hat{\delta}_C(\boldsymbol{x_2}), 
$$
for $(\boldsymbol{x_1},\boldsymbol{x_2}) \in (R^n)^2$.

For $i \in \{0,1,\ldots,n\}$, let
$$
\mathcal{D}_i(C) = \sum_{(\boldsymbol{c},\boldsymbol{c'}) \in C^2 \hbox{, } d(\boldsymbol{c},\boldsymbol{c'}) = i} \delta_{C^2}(\boldsymbol{c}, \boldsymbol{c'}),
$$
that is, $\mathcal{D}_i(C)$ is the number of ordered pairs of codewords in $C$ such that the distance between the codewords in each pair is equal to $i$.

Define $\mathcal{D}(C;x,y)$ to be the following distance enumerator in the variables $x$ and $y$:
$$
\mathcal{D}(C;x,y) = \sum_{i = 0}^n \mathcal{D}_i(C) x^{n-i}y^i.
$$
For a parity check system $(H|S)$, define the polynomial $\mathcal{N}((H|S);x,y)$ by
$$
\mathcal{N}((H|S);x,y) = 
\sum_{\boldsymbol{h} \in <\hrow{1}, \ldots,\hrow{m}>} 
\left (
|\hat{\delta}_{C}(\boldsymbol{h})|^2 x^{n - \wt{h}} y^{\wt{h}}
\right ) .
$$


\begin{theorem} \label{theorem:distance identity}
Let $(H|S)$ be a parity check system of a code $C$. Then
$$
\mathcal{D}(C;x,y) 
\quad = \quad 
\frac{1}{| R|^{n}} \mathcal{N}((H|S);\ x + (| R | - 1)y , x - y ).
$$ 
\end{theorem}
\begin{proof}
By Theorem \ref{theorem:parity check system to code} there is a coset decomposition of the code, $\decomp, $ associated with $(H|S)$ where $D_C^\perp = <\hrow{1},\ldots,\hrow{m}>$.
For any $x,y \in \mathbb{C}$, if we apply Theorem \ref{theorem:poisson} to the function 
$$
f(\boldsymbol{z},\boldsymbol{z'}) = x^{n - d(\boldsymbol{z},\boldsymbol{z'})} y^{d(\boldsymbol{z},\boldsymbol{z'})}  \hbox{ where } (\boldsymbol{z},\boldsymbol{z'}) \in (R^n)^2, 
$$
then we get that  
\begin{equation} \label{eq:distance formula}
\mathcal{D}(C;x,y) 
= \sum_{(\boldsymbol{c},\boldsymbol{c'}) \in C^2} f(\boldsymbol{c},\boldsymbol{c'}) 
= \frac{1}{| R |^{2n}} \sum_{(\boldsymbol{h}, \boldsymbol{h}') \in (<\hrow{1},\ldots,\hrow{m}>)^2} \widehat{f}(\boldsymbol{h},\boldsymbol{h'}) \widehat{\delta}_{(-C)^2}(\boldsymbol{h},\boldsymbol{h'}).
\end{equation}
By using the discrete Fourier transform (\ref{eq:fourier transform M}) we obtain the following formula for any $(\boldsymbol{h},\boldsymbol{h'}) = ((h_1,\ldots,h_n),(h_1',\ldots,h_n'))\in (<\hrow{1},\ldots,\hrow{m}>)^2$, 
\begin{align*}
\hat{f}(\boldsymbol{h},\boldsymbol{h'}) & = 
\sum_{(\boldsymbol{z},\boldsymbol{z'}) \in (R^n)^2} f(\boldsymbol{z},\boldsymbol{z'}) \echartwo{h}{h'}((-\boldsymbol{z},-\boldsymbol{z'}))\\
&  = \sum_{(\boldsymbol{z},\boldsymbol{z'}) = ((z_1,\ldots,z_n),(z_1',\ldots,z_n')) \in (R^n)^2} f(\boldsymbol{z},\boldsymbol{z'}) \prod_{i = 1}^n \varepsilon^{h_i}(-z_i) \varepsilon^{h'_i}(-z'_i) = \prod_{i = 1}^n t_i,
\end{align*}
where
$$
t_i = x \sum_{z \in R} \varepsilon^{h_i}(-z) \varepsilon^{h'_i}(-z) 
    + y \sum_{\substack{(z,z') \in R^2,\\z \neq z'}} \varepsilon^{h_i}(-z) \varepsilon^{h'_i}(-z').
$$
Since $\varepsilon^{h_i}(-z) \varepsilon^{h'_i}(-z) = \varepsilon^{h_i + h'_i}(-z)$ and $R^{\lozenge} = \{\varepsilon^0\}$, the orthogonality property (\ref{eq:Fourier orthogonality}) implies that
$$
\sum_{z \in R} \varepsilon^{h_i}(-z) \varepsilon^{h'_i}(-z) 
=
\sum_{z \in R} \varepsilon^{h_i + h'_i}(-z)
= 
\left \{
\begin{array}{ccl}
| R | & \hbox{if} & h_i' = - h_i, \\
0 & \hbox{if} & h_i' \neq - h_i. 
\end{array}
\right .
$$
Now, using that     
$$
\sum_{\substack{(z,z') \in R^2,\\z \neq z'}} \varepsilon^{h_i}(-z) \varepsilon^{h'_i}(-z')
= 
\sum_{\substack{(z,z') \in R^2}} \varepsilon^{h_i}(-z) \varepsilon^{h'_i}(-z')
-
\sum_{z \in R} \varepsilon^{h_i}(-z) \varepsilon^{h'_i}(-z) 
$$
and $(R^2)^\lozenge = \{(\varepsilon^0,\varepsilon^0)\}$ in (\ref{eq:Fourier orthogonality}), we deduce that
$$
\sum_{\substack{(z,z') \in R^2,\\z \neq z'}} \varepsilon^{h_i}(-z) \varepsilon^{h'_i}(-z') = 
\left \{
\begin{array}{ccl}
| R |^2 - | R |  & \hbox{if} & h_i' = - h_i = 0, \\
- | R |       & \hbox{if} & h_i' = - h_i \neq 0, \\
0           & \hbox{if} & h_i' \neq - h_i.
\end{array}
\right .
$$
Consequently,
$$
t_i = 
\left \{
\begin{array}{lcl}
| R | x + (| R |^2 - | R |)y & \hbox{if} & h_i' = -h_i = 0,\\
| R | x- | R | y               & \hbox{if} & h_i' = -h_i \neq 0,\\
0                    & \hbox{if} & h_i' \neq -h_i,\\
\end{array}
\right .
$$
which implies that $\hat{f}(\boldsymbol{h},\boldsymbol{h'}) = 0$ if $\boldsymbol{h'} \neq -\boldsymbol{h}$ and that
\begin{equation} \label{eq:formula_f}
\begin{array}{lcl}
\hat{f}(\boldsymbol{h},-\boldsymbol{h}) & = & (| R | x + (| R |^2 - | R |)y)^{n - \wt{h}} (| R | x - | R | y)^{\wt{h}}\\
                                        & = & | R |^n(x + (| R | - 1)y)^{n - \wt{h}} (x - y)^{\wt{h}}.
\end{array}
\end{equation}
By using that $-C$ has the coset decomposition $\bigcup_{j = 1}^s (-\boldsymbol{d}_j + D_C)$,  Theorem \ref{theorem:Fourier code} and \eqref{eq:minus_char}, we obtain that $\delta_{C}(\boldsymbol{h}) = \overline{\delta_{-C}(\boldsymbol{h})}$. Hence, again by the use of \eqref{eq:minus_char}, $\hat{\delta}_{(-C)^2}(\boldsymbol{h}, -\boldsymbol{h}) = \hat{\delta}_{-C}(\boldsymbol{h}) \hat{\delta}_{-C}(-\boldsymbol{h}) = \overline{\hat{\delta}_{C}(\boldsymbol{h})} \hat{\delta}_{C}(\boldsymbol{h}) = | \hat{\delta}_{C}(\boldsymbol{h})|^2$. By using this property and formula \eqref{eq:formula_f} in (\ref{eq:distance formula}), we obtain that
$$
\mathcal{D}(C;x,y) 
\quad = \quad 
\frac{1}{| R |^{n}} \mathcal{N}((H|S); x + (| R | - 1)y, x - y).
$$ 
The fact that the above formula holds for any $x,y \in \mathbb{C}$ concludes the proof.
\end{proof}

Let $(H|S)$ be a parity check system of a code $C$. For convenience, let $<H>$ denote the module $<\hrow{1},\ldots,\hrow{m}>$. We remark that by using Theorem \ref{theorem:Fourier parity check system} in Theorem \ref{theorem:distance identity} and use the fact that $|\sum_{j = 1}^s \varepsilon(-\srowelem{h}(j))| = |\sum_{j = 1}^s \overline{\varepsilon(\srowelem{h}(j))}| = |\sum_{j = 1}^s \varepsilon(\srowelem{h}(j))|$ we obtain that
$$
D(C;x,y) = \frac{|R|^n}{|<H>|^2} \sum_{\boldsymbol{h} \in <H>} \left ( |\sum_{j = 1}^s \varepsilon(\srowelem{h}(j))|^2 (x + (|R|-1)y)^{n - \wt{h}} (x-y)^{\wt{h}}\right).
$$

\begin{example} 
Let $C$ be the code associated to the parity check system
$$
(H|S) = 
\left (
\begin{array}{cccc|ccc}
1&1&3&5&0&1&5\\
0&4&2&2&0&2&4
\end{array}
\right )
$$
over $\mathbb{Z}_6$. By the table given in Example \ref{ex:fourier_coeff}, we obtain the following table.  
$$
\begin{array}{|c|c|c|}
\hline
w & l & |\{\boldsymbol{h} \in <\hrow{1}, \hrow{2}> : \wt{h} = w \hbox{, } \widehat{\delta_{C}}(\boldsymbol{h}) = l \}|\\
\hline
\hline
0 & 216 & 1\\
\hline
2 & 0 & 2\\
\hline
3 & 0 & 4\\
\hline
4 & -72 & 3\\
\hline
4 & 144 & 6\\
\hline
4 & 216 & 2\\
\hline
\end{array}
$$
(In the table above, there are no elements $\boldsymbol{h} \in <\hrow{1}, \hrow{2}>$ with $(\wt{h}, \delta_C(\boldsymbol{h}) = (w,l)$ for $(w,l)$-pairs not in the table.) By the use of Theorem \ref{theorem:distance identity} and the table above, we obtain the following distance polynomial enumerator $\mathcal{D}(C;x,y)$ of the code $C$,
$$
\begin{array}{rcl}
\mathcal{D}(C;x,y) &=& \frac{1}{6^4} \mathcal{N}((H|S);\ x + 5y , x - y )\\
                   &=& \frac{1}{6^4} \sum_{\boldsymbol{h} \in <\hrow{1}\hrow{2}>} \left (|\hat{\delta}_{C}(\boldsymbol{h})|^2 (x+5y)^{4 - \wt{h}} (x-y)^{\wt{h}} \right )\\
									 &=& \frac{1}{6^4} \left (216^2(x+5y)^4 + 3 \cdot 72^2 (x-y)^4 + 6 \cdot 144^2 (x-y)^4 + 2 \cdot 216^2 (x-y)^4 \right )\\
									 &=& 216 (x^4 + 30 x^2y^2 + 80 x y^3 + 105 y^4)
\end{array}
$$ 
Note that $|C| = 216$ since $\frac{|R|^n}{|<\hrow{1},\ldots, \hrow{m}>|} \cdot s = \frac{6^4}{18} \cdot 3 = 216$.
\end{example}

Any linear code over a finite commutative Frobenius ring can be represented by a parity check system $(H|\boldsymbol{0}^T)$. Hence, we remark that any parity check matrix $H'$ for a linear code over a finite commutative Frobenius ring as given in \cite{greferath04} corresponds to a parity check system $(H|\boldsymbol{0}^T)$ where $H'_{\mathrm{row}(i)} = \varphi(\hrow{i})$ and $\varphi$ is the isomorphism given in \eqref{eq:module_iso}. 

Moreover, Theorem \ref{theorem:distance identity} corresponds to the the MacWilliams identity given in \cite{greferath04} for linear codes over finite commutative Frobenius rings as follows. For any linear code $C'$, let $\mathcal{W}(C';x,y) = \sum_{i = 0}^n A_i x^{n-i}y^i$ where $A_i = |\{\boldsymbol{c'} \in C' : \wt{c'} = i\}|$. Let $(H|\boldsymbol{0}^T)$ be a parity check system of a linear code $C$. By Theorem \ref{theorem:Fourier parity check system} and since  $<\hrow{1},\ldots,\hrow{m}> = C^\perp$ it follows that
$$
|\hat{\delta}(\boldsymbol{h})|^2 = |\frac{|R|^n}{|C^\perp|} \varepsilon(0)|^2 = |C|^2
$$ 
for $\boldsymbol{h} \in C^\perp$. Consequently, by Theorem \ref{theorem:distance identity},
$$
\begin{array}{rcl}
\mathcal{W}(C;x,y) & = & \frac{\mathcal{D}(C;x,y)}{|C|}\\
                   & = & \frac{1}{|C| \cdot | R | ^{n}} \mathcal{N}((H | \boldsymbol{0}^T);\ x + (| R | - 1)y , x - y )\\
                   & = & \frac{1}{|C| \cdot | R | ^{n}} \sum_{\boldsymbol{h} \in C^\perp} \left ( |C|^2 (x + (| R | - 1)y)^{n - \wt{h}} (x-y)^{\wt{h}} \right )\\
									 & = & \frac{1}{|C^\perp|} \mathcal{W}(C^\perp; (x + (| R | - 1)y, x-y).
									 
\end{array}
$$




\begin{thebibliography}{99}
\bibitem{britz02} T. Britz,
\emph{MacWilliams identities and matroid polynomials}, Electron. J. Combin. 9 (2002), R19, 16pp.


\bibitem{delsarte72} P. Delsarte, 
\emph{Bounds  for  unrestricted  codes, by linear programming}, Philips Res. Rep. 27 (1972), 272--289. 


\bibitem{etzion98} T. Etzion and A. Vardy, 
\emph{On perfect codes and tilings, problems and solutions}, SIAM J. Discr. Math., 11(2) (1998), 205-–223.



\bibitem{greferath09} M. Greferath,
\emph{An introduction to ring-linear coding theory}, in: M. Sala, T. Mora, L. Perret, S. Sakata and C. Traverso (eds.), Gr\"obner Bases, Coding and Cryptography, Springer-Verlag, Berlin, (2009), 219--238. 

\bibitem{greferath04} M. Greferath, A. Nechaev and R. Wisbauer,
\emph{Finite quasi-Frobenius modules and linear codes}, Journ. Alg. Appl 3(3) (2004), 247--272. 

\bibitem{greferath00} M. Greferath  and S. E. Schmidt,
\emph{Finite-ring combinatorics and MacWilliams' equivalence theorem}, J. Combin. Theory Ser. A  92(1) (2000), 17--28.

\bibitem{hammons94} A. R. Hammons, Jr., P. V. Kumar, A. R. Calderbank, N. J. A. Sloane and P. Sol\'{e},
\emph{The $\mathbb{Z}_4$-linearity of Kerdock,  Preparata,  Goethals,  and  related  codes}, IEEE  Trans. Inf. Theory 40(2) (1994), 301--319. 


\bibitem{heden06} O. Heden,
\emph{A full rank perfect code of length 31}, Des. Codes Cryptogr. 38(1) (2006), 125--129.

\bibitem{heden08a} O. Heden,  
\emph{Perfect codes from the dual point of view I}, Discr. Math. 308(24) (2008), 6141--6156.

\bibitem{heden08b} O. Heden,
\emph{On perfect p-ary codes of length p+1}, Des. Codes Cryptogr. 46(1) (2008), 45--56.




\bibitem{hessler06} M. Hessler,
\emph{Perfect codes as isomorphic spaces}, Discr. Math. 306(16) (2006), 1981--1987.


\bibitem{honold01} T. Honold,
\emph{Characterization of finite Frobenius rings}, Arch. Math. 76(6) (2001), 406--415.

\bibitem{honold99} T. Honold and A. A. Nechaev,
\emph{Weighted modules and linear representations of codes}, Probl. Inf. Transm. 35(3) (1999), 205--223.

\bibitem{honold01a} T. Honold and I. Landjev,
\emph{MacWilliams identities for linear codes over finite Frobenius rings}, in: D. Jungnickel and H. Niederreiter  (eds.), Finite Fields and Applications, Springer-Verlag, Berlin, (2001), 276--292.



\bibitem{macwilliams63} F. J. MacWilliams, 
\emph{A theorem on the distribution of weights in a systematic code}, Bell Sys. Tech. J. 42 (1963), 79--94.

\bibitem{macwilliams77} F. J. MacWilliams and N. J. A. Sloane,
\emph{The Theory of Error-Correcting Codes}, North-Holland, Amsterdam, 1977.

\bibitem{nechaev91} A. A. Nechaev,
\emph{Kerdock codes in cyclic codes}, Discrete Math. Appl. 1(4) (1991), 365--384.



\bibitem{sharp00} R. Y. Sharp, 
\emph{Steps in commutative algebra}, second ed., London Mathematical Society Student Texts, vol. 51, Cambridge University Press, Cambridge, 2000.


\bibitem{terras99} A. Terras,  
\emph{Fourier analysis on finite groups and applications}, London Mathematical Society Student Texts, vol. 43, Cambridge University Press, Cambridge, 1999.

\bibitem{villanueva09} M. Villanueva,
\emph{Codis no lineals en Magma: construcci\'{o} de codis perfectes}, UAB, 2009.

\bibitem{villanueva14} M. Villanueva, F. Zeng and J. Pujol,
\emph{Efficient representation of binary nonlinear codes: constructions and minimum distance computation}, Des. Codes Cryptogr. 76(1) (2015), 3--21.

\bibitem{wood99} J. A. Wood, 
\emph{Duality for modules over finite rings and applications to coding theory}, Amer. J. Math. 121(3) (1999), 555--575.

\bibitem{wood08} J. A. Wood, 
\emph{Code equivalence characterizes finite Frobenius rings}, Amer. J. Math. 136(2) (2008), 699--706.



\end{thebibliography}
\end{document}